\pdfoutput=1
\documentclass[journal, 10pt,onecolumn]{IEEEtran}
\IEEEoverridecommandlockouts                              

\ifCLASSINFOpdf
  \usepackage[pdftex]{graphicx}
  \graphicspath{{../pdf/}{../jpeg/}}
   \DeclareGraphicsExtensions{.pdf,.jpeg,.png}
\else
  \usepackage[dvips]{graphicx}
   \graphicspath{{../eps/}}
   \DeclareGraphicsExtensions{.eps}
\fi

\usepackage{epsfig} 

\usepackage{amsmath,amssymb,array,graphicx,verbatim}

\usepackage{amsthm}

\usepackage{relsize}
\usepackage{mathtools}
\usepackage{enumerate}
\usepackage[shortlabels]{enumitem}

\newtheorem{theorem}{Theorem}
\newtheorem{lemma}{Lemma}
\newtheorem{problem}{Problem}

\newtheorem{definition}{Definition}

\newtheorem{remark}{Remark}
\newtheorem{property}{Property}

\usepackage{enumerate}
\usepackage{epstopdf}
\usepackage[noadjust]{cite}
  
\usepackage{dblfloatfix}
\usepackage{dsfont}
\usepackage{mathrsfs}

\usepackage{url}

\usepackage{tikz}
\usepackage{circuitikz}
\usepackage{tikz-cd}
\usetikzlibrary{arrows,shapes,calc,positioning}
\tikzstyle{block} = [draw,rectangle, rounded corners, minimum width=1cm, minimum height=0.8cm,text centered, line width=2pt ]
\tikzstyle{arrow} = [thick,->,>=stealth,line width=2pt]

\tikzset{cross/.style={cross out, draw=black, minimum size=2*(#1-\pgflinewidth), inner sep=0pt, outer sep=0pt},
cross/.default={1pt}}

\usepackage{enumerate}
\tikzset{
  shift left/.style ={commutative diagrams/shift left={#1}},
  shift right/.style={commutative diagrams/shift right={#1}}
}
\usetikzlibrary{calc}

\newcommand\rsmraise[1]{%
  \ifx#1\displaystyle .8\else
    \ifx#1\textstyle .8\else
      \ifx#1\scriptstyle .6\else
        .45%
      \fi
    \fi
  \fi}

\usepackage{epstopdf}
\usetikzlibrary{shapes}
\tikzstyle{block} = [draw,rectangle, rounded corners, minimum width=1cm, minimum height=0.8cm,text centered, line width=2pt ]
\tikzstyle{arrow} = [thick,->,>=stealth,line width=2pt]

\usetikzlibrary{arrows,decorations.markings,decorations}
\tikzset{
    addarrow/.style={decoration={markings, mark=at position 1 with {\arrow{stealth}}},
                     postaction={decorate}}
}

\usetikzlibrary{automata,arrows,positioning,calc}

\graphicspath{{./Sim_results/}}

\newcommand{\rg}[1]{\left[\left[#1\right]\right]}		
\newcommand{\rgc}[2]{\left[\left[#1\,\middle\vert\,#2\right]\right]}		

\begin{document}

\title{
 Worst-case Guarantees for Remote Estimation of an Uncertain Source
}

\author{Mukul Gagrani, Yi Ouyang, Mohammad Rasouli and Ashutosh Nayyar
\thanks{This work was supported by NSF Grant ECCS 1509812, CNS 1446901 and ECCS 1750041.}
}

\date{\today}
\maketitle
%

\begin{abstract}
Consider a remote estimation problem where a sensor wants to communicate the state of an uncertain source to a remote estimator over a finite time horizon. The uncertain source is modeled as an autoregressive process with bounded noise.  Given that the sensor has a limited communication budget, the sensor must decide when to transmit the state to the estimator who has to produce real-time estimates of the source state. In this paper, we consider the problem of finding a scheduling strategy for the sensor and an estimation strategy for the estimator to jointly minimize the worst-case maximum instantaneous estimation error over the time horizon. This leads to a decentralized minimax decision-making problem. We obtain a complete characterization of optimal strategies for this decentralized minimax problem. In particular, we show that an open loop communication scheduling strategy is optimal and the optimal estimate depends only on the most recently received sensor observation.  

\end{abstract}

\section{Introduction}
\label{sec:intro}


Information collection is essential for most engineering systems. 
In many applications, sensors are deployed to collect and send information to a base station/control center to estimate or control the state of the system.
In environmental monitoring, for example, remote sensors are used to measure environmental variables such as temperature, rainfall, soil moisture, etc.
The sensors collect information and transmit it to the base station through wireless communication. For a sensor with limited battery, 
the energy spent in communication is a significant factor determining the battery lifespan.
Since battery replacement is expensive for remote sensors, it is important for  sensors to adopt a transmission schedule that preserves energy while achieving a desired level of estimation accuracy.
Similar scenarios of remote estimation also arise in other applications such as  smart grids, networked control systems and healthcare monitoring \cite{ap1,ap2,kiran2014adaptive}.

The remote estimation problem with one sensor and one estimator has been studied under two different communication models:
 i) \emph{Remote estimation with pull communication protocol}: In this class of problems the estimator decides when to get data from the sensor. Since the estimator is the only decision-maker in the system, this protocol leads to a centralized sequential decision-making problem. Instances of such problems have been studied in \cite{athans,baras,wu,naghshvar}.
ii) \emph{Remote estimation with push communication protocol}: Here, the sensor makes the decision about when to send data to the estimator. The estimator decides, at each time, what estimate to produce. This leads to a decentralized decision-making problem with the sensor and the estimator as the two decsion-makers. Computing jointly optimal scheduling and estimation strategies in a decentralized setup is difficult in general. However, several works have addressed this problem by placing some restrictions on the transmission/estimation strategies and/or by making certain assumptions about the source statistics. For example, \cite{imer2005optimal,imer2010optimal} studied the problem of remote estimation under limited number of transmissions when the state process is i.i.d. and the transmission strategy is restricted to be threshold-based.  A continuous-time version of the problem is considered in  \cite{rabi2012adaptive} with a Markov state process, limited number of transmissions and a fixed estimation strategy.  \cite{xu2004optimal} derived the optimal communication schedule assuming a Kalman-like estimator. Jointly optimal scheduling and estimation strategies were derived in \cite{lipsa,nayyar2013optimal,chakravorty} for Markov sources that satisfied certain symmetry assumptions on their probability distributions. 

The uncertainties in all the aforementioned work are modeled as random variables and the objective is to minimize the expected sum cost over a finite time horizon. However, in many applications, there is no statistical model for the system variables of interest. Furthermore, guarantees on estimation accuracy at each time instant may be critical for safety concerned systems such as healthcare monitoring. For example, while monitoring the heartbeat of a patient it is desirable that the estimation error at each time is minimal. 

In this paper, we consider an uncertain source that can be modeled as a discrete-time autoregressive process with bounded noise. The source is observed by a sensor with limited communication budget. The sensor can communicate with a remote estimator that needs to produce real-time estimates of the source state.  Given such a model, we are interested in the worst-case guarantee on estimation error at any time that can be achieved under a limited communication budget. Put another way, we want to find the minimum communication budget needed to ensure that the worst-case estimation error at any time is below a given threshold.  In order to address these questions, we consider a minimax formulation of the remote estimation problem. Our goal is to design a communication scheduling strategy for the sensor and an estimation strategy for the estimator to jointly minimize the worst-case instantaneous estimation cost over all realizations of the source process.

Centralized decision and control problems where the goal is to minimize a worst-case cost have long been studied in the literature. One prominent line of work has focused on developing dynamic program type approaches for minimax problems \cite{bertsekas1973,witsenhausen1966,witsenhausen1968,baras1994robust,coraluppi1999risk}. These centralized minimax dynamic programs use analogues of stochastic dynamic programming concepts such as information states and value functions. The centralized minimax dynamic program can be interpreted in terms of a zero-sum game between the controller and an adversary who selects the disturbances to maximize the cost metric \cite{bertsekas1973,witsenhausen1966}. Dynamic games based approaches for minimax design problems were also studied in \cite{basar1}. Minimax problems where the goal is to minimize the worst-case maximum instantaneous cost were studied in \cite{bernhard1995expected,bernhard2000,bernhard2003minimax}.

In the centralized minimax problems described above, the uncertainties are described in terms of the set of values they can take. In contrast, some minimax problems have looked at systems with stochastic uncertainties. In these problems, the parameters of the stochastic uncertainties are ambiguous. These parameters are either fixed apriori but unknown or they are chosen dynamically by an adversary. In either case, the objective of the control problem is to minimize the maximum \emph{expected} cost corresponding to the worst-choice of unknown parameters. Examples of this line of work include \cite{gonzalez2002, satia1973,Iyengar2005,Wiesemann2013, osogami2015}.

Our minimax problem is most closely related to the minimax control problems studied in \cite{bertsekas1973} and \cite{bernhard2000,bernhard2003minimax}. The minimax problems in \cite{bertsekas1973} and \cite{bernhard2000,bernhard2003minimax} were centralized decision-making problem involving a single decision-maker acting over time. In contrast, our minimax problem involves two decision-makers, the sensor and the estimator, making decisions based on different information. The decentralized nature of our decision problem creates issues such as signaling where decision-makers may communicate implicitly through their actions. A decision to not communicate by the sensor, for example, can implicitly convey some information about the source to the estimator. Such signaling effects are a key reason why the joint optimization of strategies becomes a difficult problem \cite{lipsa,nayyar2013optimal}. A class of decentralized minimax control problems with partial history sharing were investigated in \cite{gagrani}. 

In order to jointly optimize the strategies for the sensor and the estimator while taking into account the signaling between them, we extend the coordinator-based approach of \cite{nayyar2013decentralized}, \cite{nayyar2013optimal} which was developed for a stochastic model and expected cost criterion to our minimax setting. Using this, we explicitly identify optimal communication scheduling and estimation strategy for our minimax problem.

\emph{Organization}: We start with a general centralized minimax control problem in Section \ref{sec:central} and then formulate the minimax remote estimation problem in Section \ref{sec:problem_formulation}. We formulate an equivalent centralized minimax control problem in Section \ref{sec:coord} and derive the optimal scheduling and estimation strategies in Section \ref{sec:opt}. We conclude in Section \ref{sec:conclusion}. 

\emph{Notation and Uncertain Variables}:
 $X_{a:b}$ denotes the collection of variables $(X_{a},X_{a+1},\ldots,X_{b})$. $\mathbb{I}_A$ denotes the indicator function of an event $A$.

We now review the concept of uncertain variables as defined in \cite{nair}. An uncertain variable is a mapping from some underlying sample space $\Omega$ to a space of interest.
We use capital letters to denote uncertain variables while small letters denote their realizations and script letters denote the spaces of all possible realizations. For example, an uncertain variable $X$  has a realization $ X(\omega) =x\in \mathcal X$ for an outcome $\omega \in \Omega$. 

Instead of probability measures as in the case of random variables, uncertain variables can be analyzed using their ranges.
The range of $X$ is defined as $\rg{X} = \{X(\omega): \omega \in \Omega\}$. Similarly, for a collection of uncertain variables $X_1,\ldots,X_n$, $\rg{X_1,\ldots,X_n} = \{(X_1(\omega),\ldots,X_n(\omega)): \omega \in \Omega\}$.
The conditional range of $X$ given $Y=y$ is denoted by $\rgc{X}{y}$ (or $\rgc{X}{Y=y}$) and is defined as  $\{X(w): Y(w)=y, w \in  \Omega\}$. We also define the uncertain conditional range $\rgc{X}{Y}$ as an uncertain variable 
that takes the value $\rgc{X}{y}$ when $Y$ takes the value $y$. 

Using the ranges of uncertain variables, an analogue of statistical independence can be defined as follows [Definition 2.1 \cite{nair}].
\begin{definition}
Uncertain variables $X_1,X_2,\dots,X_n$ are \emph{unrelated} if 
\begin{align}
\rg{X_1,\dots,X_n} = \rg{X_1}\times \dots \times \rg{X_n}
\end{align}
where $\times$ is the Cartesian product.
\end{definition}

The following property comes from the definition of unrelated uncertain variables [Lemma 2.1 \cite{nair}].

\begin{property}
\label{prop:cond_indep}
If $X$ is unrelated to $(Y,Z)$, that is, $\rg{X,Y,Z} = \rg{X} \times \rg{Y,Z}$, then
\begin{align}
\rg{X,Y|Z} = \rg{X} \times \rg{Y|Z}. 
\end{align}
\end{property}

For a function $f(x)$, we define $\sup_X f(X) := \sup_{x \in \rg{X}} f(x)$ to denote its supremum over the range of $X$. Similarly $\sup_{X_{1:n}} f(X_{1:n}) := \sup_{x_{1:n} \in \rg{X_{1:n}}} f(x_{1:n})$. Also,  $\sup_{X|y} f(X) := \sup_{x \in \rgc{X}{y}} f(x)$ denotes the supremum of $f(x)$ over the conditional range.
For a bivariate function $f(x,y)$, we have the following property.
\begin{property}
\label{prop:smoothing}
If $X,Y$ are uncertain variables. Then
\begin{align}
\sup_{X,Y} f(X,Y)
=\sup_{x \in \rg{X}} \sup_{y \in \rg{Y|x}} f(x,y).
\end{align}
Furthermore, if $Z$ is another uncertain variable, then
\begin{align}
\sup_{(X,Y)|z} f(X,Y)
=\sup_{x \in \rg{X|z}} \sup_{y \in \rg{Y|x,z}} f(x,y).
\end{align}
\end{property}
Note that the above property is the analogue of the tower property of conditional expectation with supremum playing the role of expectation.

\section{Minimax Control with Maximum Instantaneous Cost Objective}
\label{sec:central}

Consider a discrete time system with state $S_t \in \mathcal{S}$ and observation $O_t \in \mathcal{O}$ evolving according to the following dynamics:
\begin{align}
& S_{t+1} = f_{t+1}(S_t,A_t,N_{t+1}),
\\
& O_{t+1} = h_{t+1}(S_{t},A_t,N_{t+1}),
\end{align}
where  $A_t$ is the control action, $N_t$ is the noise, $t \in \mathcal{T} = \{1,2,\dots,T\}$, $S_1 = N_1$  and $O_1 =h_1(S_1)$.
The noise process  $N=\{N_t, t=1,\ldots, T\}$ is a sequence of unrelated uncertain variables. We assume that the state has two components, $S_t = (S_t^h,S_t^o)$, where $S_t^h$ is the hidden part and $S_t^o \in \mathcal{S}^o$ is the observable part.

At each time $t$, the controller's available information is  $Q_t=(O_{1:t},S^o_{1:t},A_{1:t-1})$. Note that $Q_t$ includes the history of observations $O_{1:t}$, the history of observable part of the states $S^o_{1:t}$ and the past control actions $A_{1:t-1}$. $\mathcal Q_t$ denotes the set of all possible values of $Q_t$.
The set of available control actions at $t$, which may depend on the directly observable state $S^o_t$,  is $\mathcal A(S^o_t)$.
Based on the available information at $t$, the controller takes a control action according to a function $\eta_t: \mathcal Q_t \mapsto  \mathcal A(S^o_t)$ as 
\begin{align}
& A_t = \eta_t(Q_t).
\end{align}
We call $\eta=(\eta_1,\eta_2,\dots,\eta_T)$ a strategy of the controller. The instantaneous cost at time $t$ is $\rho_t(S_t,A_t)$.
The minimax control objective is to find a strategy $\eta$ that minimizes the worst-case maximum instantaneous cost. Thus, the strategy optimization problem is 
\begin{align}
 \inf_{\eta}\Big\{\sup_{N_{1:T}} \max_{t\in\mathcal{T}} \rho_t(S_t,A_t) \Big\}.
\end{align}


Let $\Pi_t = [[S_{t}^h|Q_t]]$ be the conditional range of the hidden part of the state $S_t^h$ given the available information $Q_t$. Let $\mathcal{B}$ denote the space of all possible $\Pi_t$. 
Note that $S^o_t$ belongs to $Q_t$, so conditional range of $S^o_t$ given $Q_t$ is the singleton set i.e. $[[S_t^o | Q_t]] = \{S^o_t\}$. 

The conditional range $\Pi_t $ along with $S_t^o$ can be used as an information state for decision-making in the minimax control problem. In particular, we can obtain the following dynamic programming result using arguments from \cite{bernhard2003minimax}. 

\begin{theorem}
\label{thm:centralized}
For each $t \in \mathcal{T}$, define functions $V^*_t: \mathcal{B} \times \mathcal{S}^o \mapsto \mathbb{R}$ as follows:\\
i) For $\pi_T \in \mathcal{B}, s_T^o \in \mathcal{S}^o$,
\begin{align}
V^*_T(\pi_T,s_T^o) &:= \inf_{a_T \in \mathcal A(s_T^o)} \sup_{s_T^h \in \pi_T} \rho_T((s_T^h,s_T^o),a_T) , \label{dp:centralized_1}
\end{align}
ii) For $t<T, \pi_t \in \mathcal{B}, s_t^o \in \mathcal{S}^o$,
\begin{align}
&V^*_t(\pi_t,s_t^o) \notag \\
&:= \inf_{a_t \in \mathcal{A} (s_t^o)}\Big\{\displaystyle\sup_{s_t^h \in \pi_t,n_{t+1} \in \rg{N_{t+1}}}
 \max\Big(\rho_t(s_t,a_t),V^*_{t+1}(\Pi_{t+1},S_{t+1}^o)\Big)  \Big\}. \label{dp:centralized_2}
\end{align}
where $\Pi_{t+1}$ is given as follows,
\begin{align*}
\Pi_{t+1} = \{s_{t+1}^h: s_{t+1} =  &f_{t+1}(s_t,a_t,n_{t+1}),h_{t+1}(s_t,a_t,n_{t+1}) = o_{t+1} \\
&s_t = (s_t^h,s_t^o), s_t^h \in \pi_t, n_{t+1} \in \rg{N_{t+1}} \}.
\end{align*}

If the infimum in \eqref{dp:centralized_1}, \eqref{dp:centralized_2} is achieved, then for each $\pi_t \in \mathcal{B}$ and $s_t^o \in \mathcal{S}^o$ the minimizing $a_t$ in \eqref{dp:centralized_1}-\eqref{dp:centralized_2} gives the optimal action at time $t$ for $t \in \mathcal{T}$. Moreover, the optimal cost is given by $\sup_{Q_1}V^*_1(\Pi_1,S_1^o)$.
\end{theorem}

\begin{proof}
See Appendix \ref{app:pf_centralized}
\end{proof}

\section{Problem Formulation}
\label{sec:problem_formulation}



\begin{figure}[h]
\begin{center}
\begin{tikzpicture}

\node (Xt) at (-3.5,0){$X_t$};

\node [rectangle,draw,minimum width=1.8cm,minimum height=1cm,line width=1pt,rounded corners]
		at (-2,0) (T) {Sensor}; 
\node [rectangle,draw,minimum width=1.8cm,minimum height=1cm,line width=1pt,rounded corners]
		at (2.5,0) (R) {Estimator}; 

\node (hXt) at (4,0){$\hat X_t$};
		
\coordinate (midpoint) at (1,0);
\node[above] at (midpoint) {$Y_t$};

\draw [line width=1pt] (T.east) to[cspst, l=$U_t$] (midpoint);
\draw [thick,->,>=stealth,line width=1pt] (midpoint) -- (R.west);

\draw [thick,->,>=stealth,line width=1pt] (Xt)--(T) ;
\draw [thick,->,>=stealth,line width=1pt] (R) -- (hXt)  ;

\end{tikzpicture}
\end{center}
\caption{Remote Estimation setup}
\label{fig:systemmodel}
\end{figure}
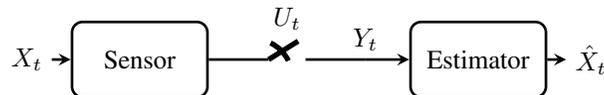

Consider a communication problem between a sensor (transmitter) and an estimator (receiver) over a finite time horizon $\mathcal T = \{1,2,\dots,T\}$, $T \geq 1$. The sensor perfectly observes a discrete-time uncertain process $X_t \in \mathbb{R}^n$ which evolves according to the following dynamics
\begin{equation}\label{eq:Xt}
X_{t+1} = \lambda A X_t + N_{t+1},
\end{equation}
where $\lambda$ is a scalar and $A$ is an orthogonal matrix. $N_t$ is an uncertain variable which lies in the ball of radius $a_t$ around the origin i.e. $||N_t|| \leq a_t$. We assume that the initial state $X_1 = N_1$. The numbers $a_1,\ldots,a_T$ are finite. Since all the noise in the system is bounded, the state $X_t$ also remains bounded for all $t$. Let $\mathcal{X} \subset \mathbb{R}^n$ denote a bounded set such that $X_t \in \mathcal{X}$ for $t \in \mathcal{T}$. 

The sensor can send the observed state to the estimator through a perfect channel. However, each transmission consumes one unit of sensor's energy, and the sensor has a limited energy budget of  $K$ units\footnote{$K$ is a fixed known integer and not an uncertain variable} with $1 \leq K < T$. Let $E_t$ denote the energy available at time $t$. We use $U_t$ to denote the transmission decision at time $t$. $U_t$ is $1$ if the current state observation is transmitted and $0$ otherwise. Note that $U_t \in \mathcal{U}(E_t)$ where $\mathcal{U}(E_t) = \{0,1\}$ if $E_t > 0 $ and $\mathcal{U}(0) = \{0\}$ i.e. there can be no transmission at time $t$ if $E_t=0$. The energy at time $t+1$ can be written as:
\begin{equation}\label{eq:Et}
E_{t+1} = \max (E_t - U_t, 0).
\end{equation}

The estimator receives $Y_t$ at time $t$ which is given as,
\begin{equation}\label{eq:Yt}
Y_t = h(X_t,U_t) =  \begin{cases}
			X_t & \quad \text{if} \quad U_t = 1, \\
			\epsilon & \quad \text{if} \quad U_t = 0, \\ 
		  \end{cases}
\end{equation}
where $\epsilon$ denotes no transmission. 
The sensor makes the transmission decision at $t$ based on available information $X_{1:t}, E_{1:t}, Y_{1:t-1}$,
\begin{align} \label{eq: U_t}
U_t = f_t(X_{1:t},E_{1:t},Y_{1:t-1}),
\end{align}  
where $f_t$ is the transmission strategy of the sensor at time $t$.  We call the collection $\mathbf{f} = (f_1,f_2,\ldots,f_T)$ the transmission strategy. \\
The estimator produces an estimate of the state $\hat{X}_t$ based on its received information $Y_{1:t}$ at time $t$ as follows:

\begin{equation}\label{eq: x_hat}
\hat{X}_t = g_t (Y_{1:t}), 
\end{equation}
where $g_t$ denotes the estimation strategy at time $t$. The collection $\mathbf{g} = (g_1,g_2,\ldots,g_T)$ is referred to as the estimation strategy.
The cost incurred under a transmission strategy $\mathbf{f}$ and estimation strategy $\mathbf{g}$ is the worst case maximum instantaneous distortion cost over the entire horizon, given by,
\begin{equation}\label{Ctotal}
J(\mathbf f,\mathbf{g}) =   \displaystyle \sup_{N_{1:T}} \max_{ t \in \mathcal T}  || X_t - \hat{X}_t ||. 
\end{equation}

We can now formulate the following problem.


\begin{problem}
\label{prob:sensor}
Determine a transmission strategy $\mathbf f$ for the sensor and an estimation strategy $\mathbf g$ for the estimator which jointly minimize the cost $J(\mathbf f,\mathbf{g})$ in \eqref{Ctotal}.
\begin{align*}
&\min_{\mathbf{f},\mathbf{g}} J(\mathbf f,\mathbf{g}) \\
& \text{subject to } \eqref{eq:Xt}-\eqref{eq: x_hat}
\end{align*}
\end{problem}

\begin{remark}
Communication scheduling and remote estimation problems similar to Problem \ref{prob:sensor} have been studied in \cite{imer2010optimal,lipsa,nayyar2013optimal}. The key differences between the problems in \cite{imer2010optimal,lipsa,nayyar2013optimal} and Problem \ref{prob:sensor} are: (i) \emph{source model}- \cite{imer2010optimal,lipsa,nayyar2013optimal} deal with a stochastic source model whereas the source model in Problem \ref{prob:sensor} is non-stochastic; (ii) \emph{objective}- \cite{imer2010optimal,lipsa,nayyar2013optimal} deal with minimizing an expected cumulative cost over a time horizon whereas the objective in Problem \ref{prob:sensor} is to minimize the worst-case instantaneous cost. The objective in Problem \ref{prob:sensor} may be more suitable for safety critical systems.
\end{remark}

Next, we provide a structural result which establishes that the sensor can ignore past values of the source and energy levels without losing performance. 
\begin{lemma}\label{lem:1}
The transmission strategy can be restricted to the form $U_t = f_t(X_t,Y_{1:t-1})$ without any loss in performance.
\end{lemma}
\begin{proof}
See Appendix \ref{sec:lem_1}.
\end{proof}

Problem \ref{prob:sensor} is a minimax sequential decision-making problem with two decision-makers (the sensor and the estimator). We will adopt the common information approach \cite{nayyar2013optimal} for  \emph{stochastic} remote estimation problem to our minimax problem. This involves formulating a  single-agent sequential decision-making problem from the perspective of an agent who knows the common information. In our setup, we can adopt the estimator's perspective to formulate the single-agent problem as done in the following section.



\section{An Equivalent Problem} \label{sec:coord}

We now formulate a new sequential decision problem that will help us to solve Problem \ref{prob:sensor}. In the new problem, we
consider the model of Section \ref{sec:problem_formulation} with the following modification. At the beginning of $t^{th}$ time step, the estimator selects a mapping $\Gamma_t : \mathcal{X}  \mapsto \{0,1\} $. $\Gamma_t$ will be referred to as the estimator's prescription to the sensor. The sensor uses the prescription to evaluate $U_t$ as follows: 
\begin{equation}\label{eq:U_coordinator}
U_t = \Gamma_t(X_t).
\end{equation}
The estimator selects the prescription based on its available information, that is,  
\begin{equation}\label{eq:Gamma}
\Gamma_t = d_t (Y_{1:t-1}),
\end{equation}
where the function $d_t$ is referred to as the prescription strategy at time $t$.  At the end of $t^{th}$ time step, the estimator produces an estimate $\hat{X}_t$  as follows
\begin{equation}\label{eq:xhat_coordinator}
\hat{X}_t = g_t (Y_{1:t}),
\end{equation}
where $g_t$ is the estimation strategy at time $t$. The cost incurred by the prescription strategy $\mathbf{d} = (d_1,\ldots,d_T)$ and the estimation strategy $\mathbf{g}= (g_1,\ldots,g_T)$ is,
\begin{equation}
\hat{J} (\mathbf{d},\mathbf{g}) = \sup_{N_{1:T}} \max_{t \in \mathcal{T}} || X_t - \hat{X}_t ||. \nonumber
\end{equation}
We consider the following problem,
\begin{problem}
\label{prob:coordinator}
Determine a prescription strategy $\mathbf d$ and an estimation strategy $\mathbf g$  to minimize the cost 
$\hat{J}(\mathbf{d},\mathbf{g})$.

\begin{align*}
&\min_{\mathbf{d},\mathbf{g}} \hat{J}(\mathbf d,\mathbf{g}) \\
& \text{subject to } \eqref{eq:Xt}-\eqref{eq:Yt},\eqref{eq:U_coordinator}, \eqref{eq:Gamma}, \eqref{eq:xhat_coordinator}
\end{align*}

\end{problem}

In Problem \ref{prob:coordinator}, the estimator is the sole decision-maker since the sensor merely evaluates the prescription at the current source state. Problem \ref{prob:coordinator} can be shown to be equivalent to Problem \ref{prob:sensor} in a similar manner as in \cite{nayyar2013optimal} for the stochastic remote estimation problem. The main idea is that for every choice of sensor strategy $\mathbf{f}$ there exists an equivalent prescription strategy $\mathbf{d}$ and vice-versa. Since this equivalence is true for every realization of the uncertain variables $N_{1:T}$, the stochastic case argument also holds in this minimax scenario.

Problem \ref{prob:coordinator} can be seen as an instance of the minimax problem formulated in Section \ref{sec:central} as follows: 
\begin{enumerate}
\item We can imagine the system operating with $2T$ decision points by splitting each time instant into two decision points: (i) At each time $t$, before the transmission at that time the estimator decides the prescription  $\Gamma_t$; (ii) After receiving $Y_t$, the estimator decides $\hat{X}_t$. We denote this decision point by $t+$ (See Figure \ref{fig1}). 

\begin{figure}[]
            	\centerline  
                	{\includegraphics[width = 0.50\textwidth]{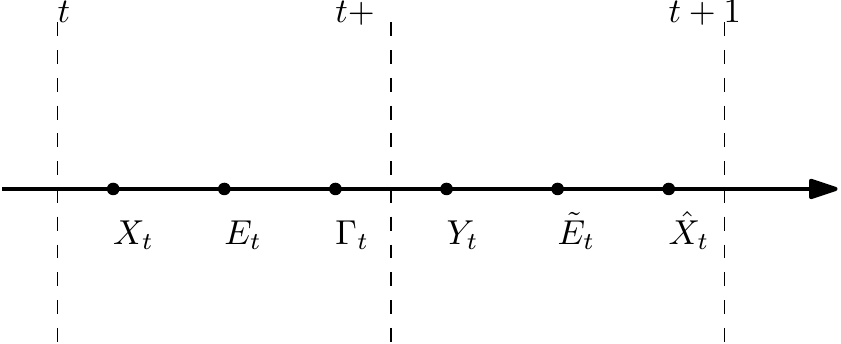}}
					\caption{Timeline of state realization, observations and actions in coordinator's problem}
            		\label{fig1}
            		
\end{figure} 

\item \textbf{State}: At $t$, the state is $S_t = (S_t^h,S_t^o) = (X_t,E_t)$ since $E_t$ is observable by the estimator. At $t+$, $S_{t+} = (S_{t+}^h,S_{t+}^o) = (X_t,\tilde{E}_t)$ where $\tilde{E}_t$ is the post-transmission energy given as
\begin{equation}\label{e_tilde}
\tilde{E}_t = E_t - \Gamma_t(X_t).
\end{equation}  

\item \textbf{Actions} : At $t$, action $A_t = \Gamma_t \in  \mathcal{A}(E_t)$, where $\mathcal{A}(E_t)$ is the collection of functions from $\mathcal{X}$ to $\mathcal{U}(E_t)$. Recall that $\mathcal{U}(0) = \{0\}$ and $\mathcal{U}(E_t) = \{0,1\}$ for $E_t>0$. At $t+$, action $A_{t+} = \hat{X}_t \in \mathbb{R}^n$. 

\item \textbf{Information}: The information available at time $t$ to choose a prescription is $Q_t = \{ Y_{1:t-1} , \Gamma_{1:t-1}, \hat{X}_{1:t-1}\}$ and at time $t+$ to generate $\hat{X}_t$ is $Q_{t+} = \{ Y_{1:t} , \Gamma_{1:t}, \hat{X}_{1:t-1} \}$.

\item \textbf{Cost}: The instantaneous cost at time $t$, $\rho_t (S_t,A_t) = 0$ and at time $t+$, $\rho_{t+}(S_{t+},A_{t+}) = ||X_t - \hat{X}_t ||$. 
\end{enumerate} 

Since Problem \ref{prob:coordinator} is an instance of the minimax problem of Section \ref{sec:central}, we can use Theorem \ref{thm:centralized}  to conclude that the optimal strategy is a function of the conditional range of the state $(X_t,E_t)$ given the estimator's information. Since $E_t$ is known to the estimator, we just need to define the conditional range of $X_t$. 
For that purpose, we define $\Theta_t$ as the pre-transmission conditional range of $X_t$  and $\Pi_t$ as the post-transmission conditional range  of $X_t$ at time $t$ as follows:
\begin{equation}
\Theta_t = [[X_t | Q_t]], \, \, \Pi_t = [[X_t | Q_{t+}]]. \nonumber
\end{equation}
The following lemma describes the evolution of the sets $\Theta_t$ and $\Pi_t$.

\begin{lemma}~\label{lem:uncertainty set}
\begin{enumerate}
\item The pre-transmission conditional range $\Theta_{t+1}$ at time $t+1$ is a function of $\Pi_t$ i.e. $\Theta_{t+1} = \phi_t (\Pi_t)$.
\item The post-transmission conditional range $\Pi_t$ is a function of $\Theta_t, \Gamma_t$ and $Y_t$  i.e. $\Pi_t = \psi (\Theta_t,\Gamma_t,Y_t)$.
\end{enumerate}
\end{lemma}
\begin{proof}~
\begin{enumerate}
\item Given the post-transmission conditional range $\Pi_t$, $\Theta_{t+1}$ is given as
\begin{align}\label{phi}
\Theta_{t+1} = \Big\{x_{t+1} : &x_{t+1} = \lambda A x_t + n_{t+1}  \nonumber \\
 & \text{for some} \, x_t \in \Pi_t ~~~ \text{and} \, || n_{t+1} || \leq a_{t+1}  \Big\}.\\
 & := \phi_t (\Pi_t) \nonumber
\end{align}
\item Given the pre-transmission conditional range $\Theta_t$,  $\Pi_t$ can be evaluated after receiving $Y_t$ as follows
\begin{align}\label{psi}
\Pi_t &= \begin{cases}
			\{Y_t\} \, &if  \, Y_t \neq \epsilon \\
			\{x_t \in \Theta_t : \Gamma_t (x_t) = 0 \} \, &otherwise
\end{cases} \\
& := \psi (\Theta_t,\Gamma_t,Y_t) \nonumber
\end{align}

\end{enumerate}
\end{proof}

Let $\mathcal{B}$ denote the space of all possible realizations of $\Pi_t, \Theta_t$ and $\mathcal{E} =\{0,1,\ldots,K\}$. Then, Theorem \ref{thm:centralized} can be used to write a dynamic program which characterizes the optimal estimates $\hat{X}_t$ and the optimal prescriptions $\Gamma_t$ in Problem \ref{prob:coordinator} as follows,

\begin{lemma}\label{lem:dp}
For $t \in \mathcal{T}$, define the functions $V_t: \mathcal{B} \times \mathcal{E}  \mapsto \mathbb{R}$ and $W_t: \mathcal{B} \times \mathcal{E}  \mapsto \mathbb{R}$ as follows: \\
(i) For $\pi_T \in \mathcal{B}$ and $\tilde{e}_T \in \mathcal{E}$ define\footnote{$\tilde{e}_t$ denotes a realization of the post-transmission energy as defined in \eqref{e_tilde}.},
\begin{equation}
 V_T (\pi_T,\tilde{e}_T) := \inf_{\hat{x}_T \in \mathbb{R}^n} \sup_{x_T \in \pi_T} || x_T-\hat{x}_T ||, \label{v_T}
\end{equation}
(ii) For $t \in \mathcal{T}$, $\theta_t \in \mathcal{B}$ and $e_t \in \mathcal{E}$ define,
\begin{align}
&W_t (\theta_t, e_t) := \inf_{\gamma_t \in \mathcal{A}(e_t)} \sup_{x_t \in \theta_t} V_t(\psi(\theta_t, \gamma_t, y_t), e_t - \gamma_t(x_t)), \label{w_t} 
\end{align}
~~ where $y_t = h(x_t,\gamma_t(x_t))$. \\
(iii) For $t < T$, $\pi_t \in \mathcal{B}$ and $\tilde{e}_t \in \mathcal{E}$ define,
\begin{align}
&V_t (\pi_t, \tilde{e}_t) := \inf_{\hat{x_t}\in \mathbb{R}^n} \sup_{x_t \in \pi_t} \lbrace \max \left( ||x_t - \hat{x}_t ||, W_{t+1}(\phi_t(\pi_t),\tilde{e}_t) \right) \rbrace.  \label{v_t} 
\end{align}
Suppose the infimum in \eqref{v_T},\eqref{w_t},\eqref{v_t} are always achieved. Then, for each $\theta_t \in \mathcal{B}$ and $e_t \in \mathcal{E}$ the minimizing $\gamma_t$ in \eqref{w_t} gives the optimal prescription at time $t$. Also, for each $\pi_t$(or $\pi_T$)$\in \mathcal{B}$, the minimizing $\hat{x}_t (\text{or } \hat{x}_T)$ gives the optimal estimate. Furthermore, $W_1([[X_1]],K)$ is the optimal cost for Problem \ref{prob:coordinator}. 
\end{lemma}
\begin{proof}
The result follows by writing the dynamic program using Theorem \ref{thm:centralized}, Lemma \ref{lem:uncertainty set} and associating the function $V_t$ with the value function at time $t+$ and $W_t$ with the value function at time $t$.
\end{proof}
Note that the above dynamic program is computationally hard to solve because: i) It involves minimization over functions in \eqref{w_t} ii) The information state is the conditional range of the source state and thus can be any arbitrary subset of $\mathcal{X}$. In the next section, we will analyze the dynamic program to obtain certain properties of the value functions which will help us in identifying the structure of the optimal strategies. 

\section{Globally optimal strategies}\label{sec:opt}

%
%

We now proceed with solving the dynamic program of Lemma \ref{lem:dp}. We proceed in four steps. 
~\\
\textbf{Step 1: Nature of optimal prescriptions} \\
We define a relation $\mathbf{Q}$ between sets which will be helpful in identifying the structure of the globally optimal prescriptions. To that end, we define the radius of a set $S \subset \mathbb{R}^n$ as $r^*(S) := \inf_{x \in \mathbb{R}^n} \sup_{y \in S} ||y - x||$. The following lemma gives the relation between the radius of a set $E$ and the radius of its transformation $\phi_t(E)$ defined by \eqref{phi}.
\begin{lemma}\label{lem:diam}
Let $E \subset \mathbb{R}^n$. Then,
\begin{equation}
r^*(\phi_t(E)) = |\lambda | r^*(E) + a_{t+1}.
\end{equation}
\end{lemma}
\begin{proof}
See Appendix \ref{sec:lem_prop_Q}.
\end{proof}

We now define a relation $\mathbf{Q}$ between sets and a property $\mathbf{Q}$ for functions.
\begin{definition}~
\begin{enumerate}
\item Let $G,H \subset \mathbb{R}^n$ be two sets. We say $G \mathbf{Q} H$ if $r^*(G)= r^*(H)$.
\item We say that a function $f:\mathcal{B} \times \mathcal{E} \mapsto \mathbb{R}$ satisfies property $\mathbf{Q}$ if 
\begin{equation}
G \mathbf{Q} H \implies f(G,e) = f(H,e) \, \, \forall e \in \mathcal{E}. \nonumber
\end{equation}
\end{enumerate}
\end{definition}

%

Let $\gamma^{all}$  denote the 'always transmit' prescription, i.e. $\gamma^{all} (x) = 1, \forall x \in \mathcal{X}$. Let $\gamma^{none}$ denote the 'never transmit' prescription, i.e. $\gamma^{none}(x) = 0, \forall x \in \mathcal{X}$.



\begin{lemma}~\label{lem:prop Q}
\begin{enumerate}
\item For each $t \in \mathcal{T}$, the functions $V_t$ and $W_t$  of Lemma \ref{lem:dp} satisfy property $\mathbf{Q}$.
\item For each $t \in \mathcal{T}$, either $\gamma^{all}$ or $\gamma^{none}$ is an optimal choice of prescription $\gamma_t$ in \eqref{w_t}.
\end{enumerate}
\end{lemma}
\begin{proof}
See Appendix \ref{sec:lem_prop_Q}.
\end{proof}

Consider two singleton sets $\{x^1_t\}$ and $\{x^2_t\}$. The first part of Lemma \ref{lem:prop Q} implies that $V_t (\{x^1_t\},e_t) = V_t (\{x^2_t\},e_t)$  because $\{x_t^1\} \mathbf{Q} \{x_t^2\}$. Thus, $V_t (\{x_t\},e_t)$ does not depend on the value of $x_t$ and can be represented as function of energy alone, that is,  $V_t (\{x_t\},e_t) = K_t(e_t)$. The second part of Lemma \ref{lem:prop Q} implies that we can replace the infimum in \eqref{w_t} by minimzation over just two prescriptions, $\gamma^{all}$ and $\gamma^{none}$. 
Using the above observations, we can reduce the dynamic program of Lemma \ref{lem:dp} to the following:
\begin{align}
&V_T (\pi_T,\tilde{e}_T) = r^*(\pi_T),\label{eq13} \\
&V_t (\pi_t, \tilde{e}_t) = \max \left(r^*(\pi_t) , W_{t+1}(\phi_t(\pi_t),\tilde{e}_t) \right), ~~ \mbox{for}~t <T, \label{eq15}
\end{align}
where \eqref{eq13} and \eqref{eq15} follow from the definition of $r^*(\pi_t)$ and the dynamic program in Lemma \ref{lem:dp}; for $e_t >0$, 
\begin{align}
W_t (\theta_t, e_t) &= \min\Big\{ \sup_{x_t \in \theta_t} V_t(\psi(\theta_t, \gamma^{all}, y_t), e_t - 1),
\notag \\ &~~~\sup_{x_t \in \theta_t} V_t(\psi(\theta_t, \gamma^{none}, y_t), e_t)\Big\} \notag \\
&=\min\Big\{ \sup_{x_t \in \theta_t} V_t(\{x_t\}, e_t - 1), \sup_{x_t \in \theta_t} V_t(\theta_t, e_t)\Big\}\notag \\
&= \min \{K_t(e_t-1),V_t(\theta_t,e_t) \}, ~~ t \in \mathcal{T}, \label{eq14}
\end{align}
where $K_t(e_t-1) = V_t (\{x_t\},e_t-1)$ for any $x_t$. For $e_t =0$,
\begin{align}
W_t (\theta_t, 0) &=\sup_{x_t \in \theta_t} V_t(\psi(\theta_t, \gamma^{none}, y_t), 0)  =V_t(\theta_t,0). \label{eq16}
\end{align}

\textbf{Step 2: Simplified information state} \\
We will now use property $\mathbf{Q}$ to  simplify the information state of the dynamic program. Lemma \ref{lem:prop Q} suggests that value functions $V_t,W_t$ depend only on the radius of the conditional range. Thus, we would expect that the radius of the conditional range can act as an information state of the dynamic program. This idea is formalized in the following lemma.

\begin{lemma}\label{lem:w_tilde}
Define $\tilde{V}_t : \mathbb{R}^+ \times \mathcal{E} \rightarrow \mathbb{R}$ and $\tilde{W}_t : \mathbb{R}^+ \times \mathcal{E} \rightarrow \mathbb{R}$ as follows:\\
(i) For $t = T$, $r \in \mathbb{R}^+$ and $\tilde{e} \in \mathcal{E}$,
\begin{align}
\tilde{V}_T (r, \tilde{e}) & := r,   \label{eq10}
\end{align}
(ii) For $t \in \mathcal{T}$,  $r \in \mathbb{R}^+$ and $\tilde{e} \in \mathcal{E}$,

\begin{equation} \label{eq11}
\tilde{W}_t (r,e) := \begin{cases}
\displaystyle \min (\tilde{V}_t (r,e), \tilde{V}_t(0, e -1)), &~~\mbox{if}~e>0, \\
 \tilde{V}_t (r,0) ~~ &\mbox{if}~e =0.
\end{cases}
\end{equation}
(iii) For $t <T$,
\begin{align}
\tilde{V}_t (r,\tilde{e}) & := \max \left(r, \tilde{W}_{t+1} (|\lambda | r + a_{t+1} , \tilde{e})\right). \label{eq12}
\end{align}

Then, for $t \in \mathcal{T}$,
\begin{align}
V_t(\pi_t,\tilde{e}_t) &= \tilde{V}_t (r^*(\pi_t),\tilde{e}_t), \label{eq_vtilde}\\
W_t(\theta_t,e_t) &= \tilde{W}_t (r^*(\theta_t),e_t). \label{eq_wtilde}
\end{align}
\end{lemma}

\begin{proof}
$V_T(\pi_T,\tilde{e}_T) =  \tilde{V}_T (r^*(\pi_t),\tilde{e}_t) $ follows from \eqref{eq10}, \eqref{eq13}. We then proceed by induction --- we first  show that \eqref{eq_wtilde} is true if \eqref{eq_vtilde} is true for $t$. \eqref{eq_wtilde} follows easily from \eqref{eq14},\eqref{eq16} and the induction hypothesis by noting that  $K_t(e_t-1) = V_t(\{x\},e_t-1) = \tilde{V}_t(0,e_t-1)$. Next, we show that \eqref{eq_vtilde} is true for $t$ if \eqref{eq_wtilde} is true for $t+1$. Using \eqref{eq15} and the induction hypothesis together with the fact that $r^*(\phi_t(\pi_t)) = |\lambda| r^*(\pi_t) + a_{t+1}$,  \eqref{eq_vtilde} can be easily established.
\end{proof}

We can further eliminate $\tilde{V}_t$ from \eqref{eq10}-\eqref{eq12} to obtain a recursive relation among $\tilde{W}_t$ given as:

\begin{equation} \label{eq:reduced_dp1}
\tilde{W}_T (r,e) = \begin{cases}
							 0 \, \hspace{1cm} \text{if} \, e >0, \\
							 r \,\hspace{1cm} \text{if} \, e = 0,
							 \end{cases}
\end{equation}

For $t<T$,

\begin{equation}\label{eq:reduced_dp2}
\tilde{W}_t (r,e) = \begin{cases}
\min \left\{ \max \{r, \tilde{W}_{t+1} (|\lambda | r + a_{t+1},e)\}, \right. \\
\left. ~~~~~~~~~~ \tilde{W}_{t+1}(a_{t+1},e-1) \right\}, ~\text{for}~e>0 \\
 \max \{r, \tilde{W}_{t+1} (|\lambda | r + a_{t+1},0)\}   ~\text{for}~e=0 
\end{cases} 
\end{equation}	
						 
The above equations can be seen as a reduced version of the dynamic program of Lemma \ref{lem:dp} with the radius of the conditional range and the energy level as the information state. Unlike the dynamic program of Lemma \ref{lem:dp}, however, the above dynamic program is completely deterministic, that is, it does not involve maximization over any uncertain variables. In the next step, we will connect this deterministic dynamic program to a deterministic optimal control problem and use it to identify optimal transmission strategy. \\

\textbf{Step 3: A deterministic control problem}\\
Consider a deterministic control system with state  $(X_t^d, E_t^d) \in \mathbb{R}^{+} \times \mathcal{E}$ and control action $U_t^d \in \mathcal{U}(E_t^d)$, where $\mathcal{U}(E_t^d) = \{0,1\}$ if $E_t^d > 0$ and $\mathcal{U}(0) = \{0\}$, operating for a time horizon  $T$. The dynamics of the state are as follows: 
\begin{align*}
X_{t+1}^d &=  \begin{cases} |\lambda| X_t^d + a_{t+1} \, \, \text{if} \, \, U_t = 0, \\
																	a_{t+1} \, \, \text{if} \, \, U_t =1, 
											\end{cases} \\
E_{t+1}^d &= \max(E_t^d - U_t^d,0)								
\end{align*}
with $X_1^d = a_1$ and $E_1^d = K$. The instantaneous cost is given by
\begin{equation*}
\rho(X_t^d,U_t^d) = \begin{cases} X_t^d   \; \;  \text{if} \, \, U_t = 0 \\
																	0 \; \; \text{if} \, \, U_t =1. 
							 \end{cases}										
\end{equation*}
The deterministic control problem can be stated as follows.
\begin{problem}
\label{prob:deterministic}
Determine  a control sequence $U^d_{1:T}$  to minimize the cost 
\[J^d(U^d_{1:T}) := \max_{t \in \mathcal{T}} \rho(X^d_t,U^d_t). \]
\end{problem}
We are interested in the above deterministic control problem because of the following lemma.
\begin{lemma}\label{lem:cost_equal}
The optimal cost for the original problem (i.e, Problem \ref{prob:sensor}), the  coordinator's problem (i.e, Problem \ref{prob:coordinator}) and the deterministic control problem  (i.e, Problem \ref{prob:deterministic}) are equal. That is,
\begin{equation}
\min_{\mathbf{f},\mathbf{g}} {J} (\mathbf{f},\mathbf{g}) =\min_{\mathbf{d},\mathbf{g}} \hat{J} (\mathbf{d},\mathbf{g}) = \min_{U^d_{1:T}} J^d(U^d_{1:T}).  \label{eq:prob_equal}
\end{equation}
\end{lemma}
\begin{proof}We have already discussed that Problems \ref{prob:sensor} and \ref{prob:coordinator} are equivalent, so we will focus on the second equality in \eqref{eq:prob_equal}.
Since the deterministic control problem is a special case of the minimax problem of Section \ref{sec:central},  we can use Theorem \ref{thm:centralized}  to write the following dynamic program for it:
\begin{align*}
\tilde{W}^d_T (x_T^d,e_T^d) &:= \begin{cases}
							 0 \; \;\text{if} \, e_T^d >0, \\
							 x_T^d \; \; \text{if} \, e_T^d = 0,
							 \end{cases} \\
\tilde{W}^d_t (x_t^d,e_t^d) &:= \min_{u_t^d \in \mathcal{U}(e_t^d)}\left\{ \max \left(\rho(x_t^d,u_t^d), \tilde{W}^d_{t+1} (x_{t+1}^d ,e_{t+1}^d)\right)\right\},
\end{align*}
for $t<T$; with $\tilde{W}^d_1 (a_1,K)$ being the optimal cost for Problem \ref{prob:deterministic}. \\
Comparing the above dynamic program with \eqref{eq:reduced_dp1}-\eqref{eq:reduced_dp2}, it is easy to see that $\tilde{W}^d_t (x,e) = \tilde{W}_t (x,e) , \forall x,e,t$. From Theorem \ref{thm:centralized}, the optimal cost of Problem \ref{prob:deterministic} is \[\tilde{W}^d_1 (a_1,K) = \tilde{W}_1 (a_1,K),\] which is the same as the optimal cost of Problem \ref{prob:coordinator}. 

\end{proof}

\textbf{Step 4: Optimal transmission and estimation strategies for Problem \ref{prob:sensor} } -
We can now identify optimal transmission and estimation strategies for Problem \ref{prob:sensor}. We start with the estimation strategy. We define $\tilde{X}_0 = 0$ and for $t \in \mathcal{T}$, 
\begin{equation}\label{eq8}
\tilde{X}_t = \begin{cases}
					Y_t \, & \text{if}\, U_t = 1, \\
					\lambda A \tilde{X}_{t-1} \, \, &\text{if}\, U_t = 0.
					\end{cases}
\end{equation}

\begin{lemma}\label{lem4}
In Problem \ref{prob:sensor} and Problem \ref{prob:coordinator}, the globally optimal estimation strategy is $g_t^*(Y_{1:t}) = \tilde{X}_t$, for $t \in \mathcal{T}$. 
\end{lemma}
\begin{proof}
See Appendix \ref{sec:main_theorem}. 
\end{proof}

Let $U_t^{d*}$ be an optimal open loop control sequence for Problem \ref{prob:deterministic}. Since Problem \ref{prob:deterministic} is an optimal control problem with determinstic dynamics we know that there exists such an open loop strategy and can be computed via the dynamic program. 
We can now identify the optimal strategies for Problem \ref{prob:sensor}.  
\begin{theorem}\label{thm:main}
Let $\mathbf{g}^*$ be the estimation strategy as defined in Lemma \ref{lem4} and $\mathbf{f}^*$ be defined as follows:
\begin{align*}
f^*_t(X_t,E_t,Y_{1:t-1}) = U_t^{d*}
\end{align*}
where $U_t^{d*}$ is an optimal open loop control sequence for Problem \ref{prob:deterministic}. 
Then, $(\mathbf{f}^*,\mathbf{g^*})$ are globally optimal strategies for Problem \ref{prob:sensor}.
\end{theorem}
\begin{proof}
See Appendix \ref{sec:main_theorem}. 
\end{proof}

Theorem \ref{thm:main} establishes that the globally optimal transmission strategy to minimize the worst-case instantaneous cost is an open-loop strategy that transmits at  pre-determined time instants. Thus, even though the sensor has access to the state and transmission history, this information is not used by the optimal transmission strategy.

\begin{remark}
We can compare the nature of optimal strategies in Theorem \ref{thm:main} with the optimal strategies in the stochastic remote estimation problem in \cite{nayyar2013optimal,lipsa}. The optimal estimation strategy obtained in our minimax setup is identical to the one obtained in the stochastic case considered in \cite{nayyar2013optimal,lipsa}. However, the optimal transmission strategy in \cite{nayyar2013optimal,lipsa} is a threshold-based strategy in contrast to the deterministic strategy obtained in our setup. 
\end{remark}

\subsection{Homogenous noise}
Consider the case when all the uncertain noise variables take values in the ball of same size i.e $a_t = a$ for all $t \in \mathcal{T}$. It turns out that transmitting at uniformly spaced intervals is optimal in this case as made precise in the following lemma. 

\begin{lemma}\label{lem:uniform_transmit}
Define $\Delta := \left\lceil \frac{T+1}{K+1} \right\rceil$. Then, 
\begin{enumerate}
\item The optimal cost for Problem \ref{prob:sensor} under homogenous noise model is, 
\begin{align*}
 c^*(K,T,a,\lambda) := \begin{cases}
 \frac{|\lambda|^{\Delta-1} - 1}{|\lambda| - 1} a  ~~ \text{when} ~~ |\lambda| \neq 1, \\
  (\Delta - 1)a ~~ \text{when} ~~ |\lambda| = 1
  \end{cases}
\end{align*}
\item An optimal control sequence for Problem \ref{prob:sensor} under homogenous noise model is given as follows:
\begin{align}\label{unif_policy}
U_t = \begin{cases}
		1\, \,  &\text{if}~ t \in \{\Delta,2\Delta,\ldots,K\Delta \} \cap \mathcal{T}  \\
		 0 \, \, & \text{otherwise}. 
		\end{cases}
\end{align}
\end{enumerate}
\end{lemma}
\begin{proof}
See Appendix \ref{sec:lem_uniform_transmit}.
\end{proof}

\begin{remark}
In the case of homogenous noise, it is possible that the sensor does not utilize all the $K$ available transmission opportunities under the transmission strategy $f^*$. For example, when $T=5, K=3$, the sensor will transmit only twice at $t=2,4$. Thus, the worst-case error achieved in this case would be the same even if $K = 2$. Therefore, one could also ask the following question: What is the minimum number of transmission opportunities ($K^*$) required so that the worst-case error is at most $\epsilon$? $K^*$ can be computed as follows:
\begin{equation}
K^* = \min \left\lbrace K \geq 1: c^*(K,T,a,\lambda) \leq \epsilon \right\rbrace
\end{equation} 
\end{remark}

\begin{remark}
Consider the problem where the estimator requests transmissions instead of the sensor deciding when to transmit. The cost of this problem is lower bounded by the cost of Problem \ref{prob:sensor} because the sensor has more information to make the transmission decision than the estimator. Moreover, since the optimal scheduling strategy obtained for Problem \ref{prob:sensor} is an open loop strategy, it can also be implemented in this new problem. Therefore, the results obtained for Problem \ref{prob:sensor} also hold for this problem.
\end{remark}

\begin{remark}
Consider the problem where the sensor can observe the source state only $M$ times instead of observing the state at each time with $M \geq K$. In addition to the scheduling strategy, here the sensor must also decide when to observe the source. The cost of this problem is lower bounded by the cost of Problem \ref{prob:sensor} because the sensor has less information in this case compared to Problem \ref{prob:sensor}. Also, since the optimal scheduling strategy for Problem \ref{prob:sensor} is an open loop strategy, the sensor in this problem can take observations at the fixed times when it transmits, thereby achieving the same cost as in Problem \ref{prob:sensor}. Therefore, the results obtained for Problem \ref{prob:sensor} also hold for this problem.
\end{remark}

\begin{remark}
For each $t \in \mathcal{T}$, let $\mathcal{B}_t$ be any set such that $\mathcal{B}_t$ is symmetric (i.e. if $n \in \mathcal{B}_t$ then $-n \in \mathcal{B}_t$) and $\sup_{n_t \in \mathcal{B}_t} ||n_t|| = a_t$. It can be shown that the optimal transmission and estimation strategy remains the same if the noise $N_t$ lies in the set $\mathcal{B}_t$. 
\end{remark}


\section{Conclusion}
\label{sec:conclusion}


 We considered the problem of remote estimation of a non-stochastic source over a finite time horizon where the sensor has a limited communication budget. Our objective was to find jointly optimal scheduling and estimation strategies which minimize the worst-case  maximum instantaneous estimation error over the time horizon. This problem is a decentralized  minimax decision-making problem. Our approach started with  the dynamic program (DP)  for a general centralized  minimax control problem. We framed our decentralized minimax problem from the estimator's perspective and used the common information approach to write down a dynamic program. 
This dynamic program, however, involved minimization over functions. By identifying a key property of the value functions, we were able to characterize the globally optimal strategies. In particular, we show that an open loop transmission strategy and simple Kalman-like estimator are jointly optimal. We also described related problems where the same optimal strategy holds. 



\appendices
\section{Proof of Theorem \ref{thm:centralized}}
\label{app:pf_centralized}

To prove Theorem \ref{thm:centralized}, we first derive some useful properties.
Recall that $N = (N_1,N_2,\ldots,N_T)$ is the collection of all the noise variables in the system. Note that given the strategy $\eta$, the state $S_r$ and the information $Q_r$ can be written down as some function of $N$ for $r \in \mathcal{T}$. Thus, for any function $f$ and $r\geq t$ we can write $\sup_{(S_r,Q_r) | q_t } f(S_r,Q_r) = \sup_{N|q_t}  f(S_r,Q_r)$

For any strategy $\eta$, we define its ``cost-to-go'' function at time $t$ as
\begin{align}\label{v_eta}
V^{\eta}_t(q_t) := \sup_{N|q_t} \max_{r\geq t} \rho_r(S_r,\eta_r(Q_r)),
\end{align}
which is a function of the realization $q_t$ of available information at time $t$. 
Then it is clear that the worst case cost of strategy $\eta$ is 
\begin{align}
\sup_{N} \max_{t\in\mathcal{T}} \rho_t(S_t,\eta_t(Q_t))
= \sup_{Q_1} V^{\eta}_1(Q_1) .
\end{align}

We also define the value function of the problem at $t$ to be
\begin{align}
 &V^*_T(q_T) :=  \inf_{a_T \in \mathcal{A} (s_T^o)}\Big\{\sup_{ {N}|{q_T} }\rho_t(S_T,a_T)\Big\},
\\
& V^*_t(q_t) :=  \inf_{a_t \in \mathcal{A} (s_t^o)} \left\{
\sup_{{N} |{(q_t,a_t)} } 
 \max\left(\rho_t(S_t,a_t),V^*_{t+1}(Q_{t+1})\right)
\right\}
\end{align}
We have the following result.
\begin{lemma}
\label{lm:valuefunctions}
For any strategy $\eta$, at each time $t$ and for every realization $q_t$, we have
\begin{align}
 V^*_t(q_t) \leq V^{\eta}_t(q_t).
\end{align}
\end{lemma}
\begin{proof}
The proof is done by induction. At $T$ we have
\begin{align}
 & V^*_T(q_T) 
\notag\\
 = & \inf_{a_T \in \mathcal{A} (s_T^o)}\Big\{\sup_{{N}|{q_T}}\rho_t(S_T,a_T)\Big\} \
 \notag\\
 \leq & \sup_{{N}|{q_T}}\rho_T(S_T,\eta_T(q_T))= V^{\eta}_T(q_T).
\end{align}
Suppose the lemma is true at $t+1$. Then at $t$ we have
\begin{align}
&V^{\eta}_t(q_t) \notag\\ = & \sup_{{N}|{q_t}} \max_{r\geq t} \rho_r(S_r,\eta_r(Q_r))
\notag\\
= &\sup_{{N}|{q_t}} \max\Big(\rho_t(S_t,\eta_t(q_t)), \max_{r\geq t+1} \rho_r(S_r,\eta_r(Q_r))\Big)
\notag\\
= &\max\Big(\sup_{{N}|{q_t}} \rho_t(S_t,\eta_t(q_t)), \sup_{{N}|{q_t}} \max_{r\geq t+1} \rho_r(S_r,\eta_r(Q_r))\Big)
\label{eq:pfVgamma1}
\end{align}
From Property \ref{prop:smoothing} we get
\begin{align}
& \sup_{{N}|{q_t}} \max_{r\geq t+1} \rho_r(S_r,\eta_r(Q_r)) 
\notag\\
= & \sup_{{Q_{t+1}}|{(q_t,A_t=\eta_t(q_t) )}} \left( \sup_{{N}|{(Q_{t+1},q_t,A_t=\eta_t(q_t) )}} \max_{r\geq t+1} \rho_r(S_r,\eta_r(Q_r)) \right)
\notag\\
= & \sup_{{Q_{t+1}}|{(q_t,A_t=\eta_t(q_t) )}} V^{\eta}_{t+1}(Q_{t+1}).
\label{eq:pfVgamma2}
\end{align}

Now from \eqref{eq:pfVgamma1}-\eqref{eq:pfVgamma2} and the induction hypothesis we get
\begin{align}
&V^{\eta}_t(q_t) \notag\\
= &\max\Big(\sup_{{N}|{q_t}} \rho_t(S_t,\eta_t(q_t)), \sup_{{Q_{t+1}}|{(q_t,A_t=\eta_t(q_t) )}} V^{\eta}_{t+1}(Q_{t+1})\Big)
\notag\\
\geq &\max\Big(\sup_{{N}|{q_t}} \rho_t(S_t,\eta_t(q_t)), \sup_{{Q_{t+1}}|{(q_t,A_t=\eta_t(q_t) )}} V^{*}_{t+1}(Q_{t+1})\Big)
\notag\\
\geq &V^{*}_t(q_t).
\end{align}

\end{proof}

It is straightforward to see that a strategy $\eta^*$ achieving infimum at each stage in the definition of $V^*_t(q_t)$ will be optimal and its cost will be $\sup_{{Q_1}} V^{*}_1(Q_1)$.

Let $\Theta_t =  [[S_{t}|Q_t]]$ be the conditional range of the state at time $t$. Recall that $\Pi_t =  [[S^h_{t}|Q_t]]$.  Note that $\Pi_t$ and $\Theta_t$ are related as follows
\begin{equation}\label{eq:theta}
\Theta_t =  [[S_{t}^h,S^o_t|Q_t]] = [[S_{t}^h|Q_t]] \times \{S^o_t\} = \Pi_t \times  \{S^o_t\}.
\end{equation}

The evolution of $\Theta_t$ has the following feature.
\begin{lemma}
\label{lm:piupdate}
There exists a function $\phi_t(\theta_t,a_t,o_{t+1},s^o_{t+1})$ such that
\begin{align}
\Theta_{t+1} = \phi_t(\Theta_t,A_t,O_{t+1},S^o_{t+1}).
\end{align}
\end{lemma}
\begin{proof}
We can write $(O_{t+1},S^o_{t+1}) = \tilde h_{t+1}(S_{t},A_t,N_{t+1})$ for some function $\tilde h_{t+1}$. Under any strategy $\eta$,
\begin{align}
&\Theta_{t+1} \nonumber\\
=
& \rg{S_{t+1}|Q_{t+1}}
			\notag\\
		  =& \rg{ S_{t+1}|Q_t, A_t, O_{t+1},S^o_{t+1} }
		  \notag\\
		  =& \rg{ f_{t+1}(S_t,A_t,N_{t+1})|Q_t, \tilde h_{t+1}(S_{t},A_t,N_{t+1})=(O_{t+1},S^o_{t+1})}
			\notag\\
          =& \Big\{f_{t+1}(s_t,A_t,n_{t+1}): \tilde h_{t+1}(s_{t},A_t,n_{t+1})=(O_{t+1},S^o_{t+1}),  \notag\\
          &(s_t,n_{t+1}) \in \rg{S_t,N_{t+1}|Q_t}\Big\}
          \notag\\
          =& \Big\{f_{t+1}(s_t,A_t,n_{t+1}): \tilde h_{t+1}(s_{t},A_t,n_{t+1})=(O_{t+1},S^o_{t+1}), \notag\\
          &s_t \in \rg{S_t|Q_t}, n_{t+1} \in \rg{N_{t+1}}\Big\}
\label{eq:pitplus}
\end{align}
where the last equality follows from Property \ref{prop:cond_indep} and the fact that $N_{t+1}$ is unrelated to $S_t$ and $Q_t$. Therefore, \eqref{eq:pitplus} implies that $\Theta_{t+1}$ is a function of $A_t,O_{t+1},S^o_{t+1}$ and $\Theta_t = \rg{S_t|Q_t}$.
\end{proof}

Now let's prove Theorem \ref{thm:centralized}. Its easy to observe using \eqref{eq:theta} that $\Theta_t$ can be completely characterized using $\Pi_t, S_t^o$. Thus, to prove Theorem \ref{thm:centralized}  it suffices to show that the optimal value function depends only on $\Theta_t$.

\begin{proof}[Proof of Theorem \ref{thm:centralized}]
Lemma \ref{lm:valuefunctions} ensures that optimal costs and optimal strategies are characterized   by the dynamic program
\begin{align}
& V^*_T(q_T) =  \inf_{a_T}\Big\{\sup_{s_T \in \theta_T}\rho_t(s_T,a_T)\Big\}
\\
& V^*_t(q_t) = \inf_{a_t}\Big\{
 \max(\sup_{s_t \in \theta_t}\rho_t(s_t,a_t),\sup_{{Q_{t+1}}|{(q_t,a_t)}}V^*_{t+1}(Q_{t+1}))
\Big\}
\end{align}
Therefore, it just remains to show that the above value function at $t$ can be written as a function of $\theta_t$. Then the optimal value will depend only on $\theta_t$ instead of the entire $q_t$. 
This claim about  the value functions is proved by induction.
At $T$, we have
\begin{align}
V^*_T(q_T) = & \inf_{a_T}\Big\{\sup_{s_T \in \theta_T}\rho_t(s_T,a_T)\Big\}
\notag\\
 =: &V^*_T(\theta_T).
\end{align}
Suppose this claim is true at $t+1$. 
From Lemma \ref{lm:piupdate} and the induction hypothesis we have
\begin{align}
& \sup_{{Q_{t+1}}|{(q_t,a_t)}}V^*_{t+1}(Q_{t+1}) 
\notag\\
= &\sup_{{(q_t,a_t,O_{t+1},S^o_{t+1} )}|{(q_t,a_t)}}V^*_{t+1}(\Theta_{t+1}) 
\notag\\
= &\sup_{{(O_{t+1},S^o_{t+1} )}|{(q_t,a_t)}}V^*_{t+1}(\phi_t(\theta_t,a_t,O_{t+1},S^o_{t+1})).
\end{align}
Since $(O_{t+1},S^o_{t+1}) = \tilde h_{t+1}(S_{t},A_t,N_{t+1})$ as in the proof of Lemma \ref{lm:piupdate}, the above equation can be further expressed as
\begin{align}
& \sup_{{Q_{t+1}}|{(q_t,a_t)}}V^*_{t+1}(Q_{t+1}) \notag\\
= &\sup_{{(S_t,N_{t+1} )}|{(q_t,a_t)}}V^*_{t+1}(\phi_t(\theta_t,a_t,\tilde h_{t+1}(S_{t},a_t,N_{t+1})))
\notag\\
= &\sup_{s_t \in \theta_t,n_{t+1} \in \rg{N_{t+1}}}V^*_{t+1}(\phi_t(\theta_t,a_t,\tilde h_{t+1}(s_{t},a_t,n_{t+1})))
\end{align}
where the last equality follows from Property \ref{prop:cond_indep} since $\theta_t=\rgc{S_t}{q_t,a_t}$ depends on the realization of $Q_t,A_t$ and $N_{t+1}$ is unrelated to all variables before $t+1$.
Therefore, the value function at $t$ is equal to
\begin{align}
 V^*_t(q_t) = &\inf_{a_t}
\Big\{
 \max\Big(\sup_{s_t \in \theta_t}\rho_t(s_t,a_t), \notag\\
 &\quad \sup_{s_t \in \theta_t,n_{t+1} \in \rg{N_{t+1}}}V^*_{t+1}(\phi_t(\theta_t,a_t,h_{t+1}(s_{t},a_t,n_{t+1})))\Big)
\Big\}
\notag\\
=: &  V^*_t(\theta_t)
\end{align}
which finishes the proof of the claim. 
It is straightforward to see that a strategy achieving infimum at each stage will have a cost equal to $\sup_{{Q_1}} V^{*}_1(Q_1) =\sup_{q_1 \in \rg{Q_1}} V^*_1(\theta_1)$ where $\theta_1 = \rg{S_1|q_1}$. Hence the proof is complete.
\end{proof}

%

\section{Proof of Lemma \ref{lem:1}} \label{sec:lem_1}
Fix the estimator's strategy to some arbitrary $\mathbf{g}$. Define $S_t = (X_t,E_t,Y_{1:t-1})$. Then, 
\begin{align*}
S_{t+1} &=	\left(\begin{array}{c}
X_{t+1} \\
E_{t+1} \\
Y_{1:t}
\end{array}
\right)
=\left(\begin{array}{c}
\lambda A X_t + N_t \\
\max(E_t - U_t, 0) \\
Y_{1:t-1}, h(X_t,U_t)
\end{array}
\right) \nonumber  \\ 
&=: \tilde{F}_t (S_t,U_t,N_t) .					      						
\end{align*}
The instantaneous cost at time $t$ can be written as
\begin{align*}
C_t &= ||X_t - \hat{X}_t|| = ||X_t - g_t(Y_{1:t})|| \\
&= ||X_t - g_t(Y_{1:t-1},Y_t)|| = ||X_t - g_t(Y_{1:t-1}, h(X_t,U_t))|| \\
& =: \rho(S_t,U_t).
\end{align*}
The problem of optimizing the transmission strategy is now an instance of the  centralized minimax control problem discussed in Section \ref{sec:central} with  $S_t$ as the directly observable state and $U_t$ as the action. Since there is no hidden state for the transmitter, the optimal transmission strategy at time $t$ is a function of the current state $S_t$. 

Since the above argument holds for any arbitrary estimation strategy $\mathbf{g}$, it holds true for an optimal estimation strategy as well. Therefore, it is sufficient to consider transmission strategies of the form $U_t =f_t(S_t)=f_t(X_t,E_t,Y_{1:t-1})$. Moreover, since $E_t$ can be inferred from $Y_{1:t-1}$, we can further restrict transmission strategies to the form $U_t = f_t(X_t,Y_{1:t-1})$ without any loss in performance.

\section{Proof of lemmas \ref{lem:diam} and \ref{lem:prop Q}}\label{sec:lem_prop_Q}
\subsection*{Proof of Lemma \ref{lem:diam}}
The proof is trivial if $\lambda =0$, so we will focus on the case of $\lambda \neq 0$. For a set $S$ and $x \in \mathbb{R}^n$, define $r(S,x) := \sup_{z \in S} ||z -x||$. For a fixed $x$, we can write,
\begin{align}
&r(\phi_t(\pi),x) = \sup_{z \in \phi_t(\pi)} || z - x || = \sup_{y \in \pi, ||w|| \leq a_{t+1}} ||\lambda A y + w -x ||\label{r1}  \\
&  \leq \sup_{y \in \pi} ||\lambda A y -x || + \sup_{||w|| \leq a_{t+1}} ||w||  \nonumber \\
& = |\lambda|  \sup_{y \in \pi} ||A(y - \frac{A^{-1} x}{\lambda})|| + a_{t+1} \nonumber  \\
& = |\lambda|  \sup_{y \in \pi} ||y - \tilde{x}|| + a_{t+1},  \, \text{where} \, \tilde{x} = \frac{A^{-1}x}{\lambda}, \nonumber \\
& =  |\lambda| r(\pi,\tilde{x}) + a_{t+1} \label{r2}
\end{align}
where we used the fact that for any vector $u$, $||A u|| = ||u||$ since $A$ is an orthogonal matrix. 

Let $\epsilon >0$. Then, $\exists y_{\epsilon} \in \pi$ such that $|| y_{\epsilon} - \tilde{x} || > r (\pi, \tilde{x}) - \frac{\epsilon}{|\lambda|}$. Taking $w = a_{t+1} \frac{A(y_{\epsilon} - \tilde{x})}{|| y_{\epsilon} - \tilde{x} ||} sign(\lambda)$ and $y = y_{\epsilon}$ we get  
\begin{align}
\sup_{y \in \pi, ||w|| \leq a_{t+1}}  & ||\lambda A y + w -x || \geq || \lambda A y_{\epsilon} - x + a_{t+1} \frac{A(y_{\epsilon} - \tilde{x})}{|| y_{\epsilon} - \tilde{x} ||} sign(\lambda) || \nonumber \\
& = |\lambda| ||  y_{\epsilon} - \tilde{x} + \frac{a_{t+1}}{|\lambda|} \frac{y_{\epsilon} - \tilde{x}}{|| y_{\epsilon} - \tilde{x} ||}  || \nonumber \\
&=  |\lambda| || y_{\epsilon} - \tilde{x}  || + a_{t+1}  > |\lambda| r(\pi,\tilde{x}) + a_{t+1} - \epsilon \label{r_ineq}
\end{align}
Since $\epsilon$ is arbitrary \eqref{r1} and \eqref{r_ineq} implies,
\begin{equation}\label{r_ineq2}
\implies  r (\phi_t(\pi),x) \geq  |\lambda| r(\pi,\tilde{x}) + a_{t+1} 
\end{equation}
Using \eqref{r2} and \eqref{r_ineq2} we get $ r (\phi_t(\pi),x)=  |\lambda| r(\pi,\tilde{x}) + a_{t+1}$. Thus,
\begin{equation}\label{r_phi}
r^*(\phi_t(\pi)) = \inf_x ( |\lambda| r(\pi,\frac{A^{-1}x}{\lambda}) + a_{t+1}) = |\lambda| r^*(\pi) + a_{t+1}
\end{equation}
where the second equality follows since $\frac{1}{\lambda} A^{-1}$ is invertible.

\subsection*{Proof of Lemma \ref{lem:prop Q}}

1) We start by showing that the lemma is true for $t=T$. Note that $V_T (\pi,\tilde{e}) = r^*(\pi)$ by definition of $r^*(\pi)$ and $V_T (\pi,\tilde{e})$. Therefore, it follows trivially that $V_T$ satisfies property $\mathbf{Q}$.

Now, consider two sets $\theta$ and $\tilde{\theta}$ such that $\theta \mathbf{Q} \tilde{\theta}$. At $t= T$, observe that if $e_T > 0$, the prescription $\gamma^{all}$ achieves the infimum in \eqref{w_t} and the corresponding infimum value is zero. Thus, $W_T(\theta, e) = W_T(\tilde{\theta}, e) =  0, \forall e>0$.  If $e_T = 0$ then the only possible choice of $\gamma_T$ is $\gamma^{none}$. Observe from \eqref{psi} that $\psi(\theta,\gamma^{none},\epsilon) = \theta$, thus it follows that $W_T(\theta, 0) = V_T(\theta,0)$ from \eqref{w_t}. Since $\theta \mathbf{Q} \tilde{\theta}$ and $V_T$ satisfies property $\mathbf{Q}$, we have $W_T(\theta,0) = W_T(\tilde{\theta},0)$. Thus, $W_T$ satisfies property $\mathbf{Q}$. \\

2) We now proceed by induction to prove that the lemma is true for $t<T$. We first show that if $W_{t+1}$ satisfies property $\mathbf{Q}$, then so does $V_t$. \eqref{v_t} can be simplified to the following:
\begin{equation}\label{v_t_multidim_simplified}
V_t (\pi_t, \tilde{e}_t) =  \max \lbrace r^*(\pi_t), W_{t+1}(\phi_t(\pi_t),\tilde{e}_t)) \rbrace,
\end{equation}
where $r^*(\pi_t) = \inf_{\hat{x_t}\in \mathbb{R}} \sup_{x_t \in \pi_t} ||x_t - \hat{x}_t||$. Let $\pi,\tilde{\pi}$ be two sets such that $\pi \mathbf{Q} \tilde{\pi}$. Then, $r^*(\pi_t)= r^*(\tilde{\pi}_t)$. Hence, the first term inside the maximization in \eqref{v_t_multidim_simplified} is the same for $\pi$ and $\tilde{\pi}$. \\

It follows from Lemma \ref{lem:diam} that if $\pi \mathbf{Q} \tilde{\pi}$ then $\phi_t(\pi) \mathbf{Q} \phi_t(\tilde{\pi})$. Then, $W_{t+1}(\phi_t(\pi),e_t) = W_{t+1}(\phi_t(\tilde{\pi}),e_t)$ follows using the induction hypothesis.
Thus, both the terms in the maximization in \eqref{v_t_multidim_simplified} satisfy property $\mathbf{Q}$. Therefore, $V_t$ also satisfies property $\mathbf{Q}$. \\

Next, we show that if $V_t$ satisfies property $\mathbf{Q}$ then so does $W_t$. 
Observe that if $\{x_1\}$ and $\{x_2\}$ are two singleton sets then $V_t(\{x_1\},e) = V_t(\{x_2\},e)$ since $V_t$ satisfies property $\mathbf{Q}$. Thus, we may write 
\begin{equation}\label{eq5}
V_t(\{x\},e) = K_t(e) \quad \forall x \in \mathcal{X}. \notag
\end{equation}

Let $e_t > 0$. Define $W_t^{\gamma} (\theta_t, e_t)$ for a given prescription $\gamma$ as follows:
\begin{equation}\label{eq3}
W_t^{\gamma} (\theta_t, e_t) =  \sup_{x_t \in \theta_t} V_t(\psi(\theta_t, \gamma,y_t), e_t - \gamma(x_t)).
\end{equation} 
Then, $W_t(\theta_t,e_t) = \displaystyle\inf_{\gamma} W_t^{\gamma}(\theta_t,e_t)$. 
For any prescription $\gamma$, let ${A_{\gamma,\theta_t} := \{x \in \theta_t: \gamma(x) = 0\} }$ be the set of the state values in $\theta_t$ which are mapped to the control action $0$.  If $A_{\gamma,\theta_t} = \emptyset$, then 
\begin{equation*}
W_t^{\gamma} (\theta_t, e_t) = K_t(e_t-1).
\end{equation*}
If $\theta_t \text{\textbackslash} A_{\gamma,\theta_t} = \emptyset$, then
\begin{equation*}
W_t^{\gamma} (\theta_t, e_t) = V_t(\theta_t,e_t).
\end{equation*}
If neither $A_{\gamma,\theta_t}$ or $\theta_t \text{\textbackslash} A_{\gamma,\theta_t}$ is empty, then 
\begin{equation*}\label{eq4}
\resizebox{0.45\textwidth}{!}{$W_t^{\gamma} (\theta_t, e_t) = \max \{ V_t (A_{\gamma,\theta_t},e_t), \displaystyle\sup_{x \in \theta_t \text{\textbackslash} A_{\gamma,\theta_t}} V_t (\{ x \}, e_t -1) \}$}
\end{equation*}
\begin{align*}
 &= \max \{ V_t (A_{\gamma,\theta_t},e_t), K_t(e_t -1) \}  \geq K_t(e_t-1).
\end{align*}
Also, it is easy to see that for the prescriptions $\gamma^{all}$ and $\gamma^{none}$ we have $W_t^{\gamma^{all}} (\theta_t, e_t) = K_t(e_t-1)$ and $W_t^{\gamma^{none}} (\theta_t, e_t) =V_t(\theta_t,e_t)$ respectively. 
Thus, it is clear that 
\begin{align}
W_t(\theta_t,e_t) &= \displaystyle\inf_{\gamma} W_t^{\gamma}(\theta_t,e_t) \notag \\
&=\min \{ K_t(e_t -1), V_t(\theta_t)\} \notag \\
&= \min\{ W_t^{\gamma^{all}}(\theta_t,e_t), W_t^{\gamma^{none}}(\theta_t,e_t) \}.
\end{align}
Thus, either $\gamma^{all}$ or $\gamma^{none}$ is an optimal prescription at time $t$. \\

Now, if $\theta \mathbf{Q} \tilde{\theta}$, then it follows from the induction hypothesis that $W_t(\theta,e_t) = \min \{V_t(\theta,e_t), K_t(e_t-1)\} = \min \{V_t(\tilde{\theta},e_t), K_t(e_t-1)\} = W_t(\tilde\theta,e_t) $. Similar arguments can be made if $e_t=0$. Therefore, $W_t$ satisfies property $\mathbf{Q}$. 

Thus, by induction, $V_t$ and $W_t$ satisfy property $\mathbf{Q}$ for all $t=1,2,\ldots,T$.

\section{}\label{sec:main_theorem}
\subsection*{Proof of Lemma \ref{lem4}}

We first show that the post-transmission conditional range $\Pi_t$ is a ball centered around $\tilde{X}_t$ under a globally optimal prescription strategy. This can be done by a simple induction argument: At $t=1$ one of the following two will happen
\begin{enumerate}
\item If $\gamma_1 = \gamma^{all}$, then $\tilde{X}_1 = X_1$ and $\Pi_1 = \{X_1\}$.
\item If $\gamma_1 = \gamma^{none}$, then $\tilde{X}_1 = 0$ and $\Pi_1 = \{x_1 : ||x_1|| \leq a_1\}$.
\end{enumerate}
Hence, the claim is true for $t=1$. Let the claim be true for $t$. Then, at time $t+1$ one of the following will happen,
\begin{enumerate}
\item If $\gamma_{t+1} = \gamma^{all}$, then  $\tilde{X}_{t+1} = X_{t+1}$ and $\Pi_{t+1} = \{X_{t+1}\}$. 
\item If $\gamma_{t+1} = \gamma^{none}$, then  $\tilde{X}_{t+1} = \lambda A \tilde{X}_t$. In this case, $\Pi_{t+1} = \Theta_{t+1} = \{x_{t+1} : x_{t+1} = \lambda A x_t + n_{t+1} , \, x_t \in \Pi_t, || n_{t+1} || \leq a_{t+1}\}$ i.e. $\Pi_{t+1}$ is obtained by rotating $\Pi_t$ using $A$, scaling it by $\lambda$ and then adding it to a ball centered around origin of radius $a_{t+1}$. Using the induction hypothesis that $\Pi_t$ is a ball centered at $\tilde{X}_t$, it follows that $\Pi_{t+1}$ is a ball centered at $\tilde{X}_{t+1} = \lambda A \tilde{X}_t$.  
\end{enumerate}

Thus, $\Pi_t$ is a ball centered around $\tilde{X}_t$ for all $t$. Therefore, the infimum in \eqref{v_t} will be achieved by $\tilde{X}_t$. Hence, $\tilde{X}_t$ is the optimal esimate at time $t$.

\subsection*{Proof of Theorem \ref{thm:main}}

We will argue that the strategies $\mathbf{f}^*,\mathbf{g^*}$ achieve the globally optimal cost for Problem \ref{prob:sensor}. Denote the $K$ time instants\footnote{If $t_i=t_{i+1}$ for some $i$, the controller chooses control action $1$ fewer than $K$ times.} with $U^{d*}_t$ equal to $1$ by  $1 \leq t_1 \leq \ldots t_K \leq T$ with the convention that $t_{K+1}= T+1, t_0 = 0$ and $X_0^d = 0$. 

Now, in Problem \ref{prob:deterministic}, if $t_{i}+1<t_{i+1}$, the state grows in the interval $[t_i+1,t_{i+1}-1]$ for all $i$ and in the interval $[t_K+1,T]$ if $t_K<T$. Therefore,  
\begin{equation}\label{eq:cost}
 J^d(U^{d*}_{1:T}) := \max_{t \in \mathcal{T}} \rho(X^d_t,U^{d*}_t) =  \max_{0 \leq i \leq K} X^d_{t_{i+1} -1}\mathbb{I}_{t_i+1<t_{i+1}}
\end{equation}

Using \eqref{eq:cost} and the state dynamics we can write 
\begin{align}
J^d(U^{d*}_{1:T}) = \max_{0 \leq i \leq K} & \left( \sum_{j=t_i+1}^{t_{i+1} -1} |\lambda|^{t_{i+1}-1-j} a_{ j } \right) \,  \mathbb{I}_{t_i+1<t_{i+1}} \label{T_term}
\end{align}
Now, consider the worst case instantaneous cost in Problem \ref{prob:sensor} under the strategy $\mathbf{f}^*,\mathbf{g^*}$. First consider the interval $[1,t_1]$. If $t_1=1$ then the estimation error is $0$ in this interval. When $t_1>1$, let $1 \leq t < t_1$, then $\hat{X}_t = 0$ under $\mathbf{g^*}$. Then at time $t$, the worst case estimation error is $\sup_{N_{1:t}} ||\sum_{j=0}^{t-1} \lambda^j A^j N_{t-j} || = \sum_{j=1}^t |\lambda|^{t-j} a_j$. Hence, the worst case estimation error in $[1,t_1]$ is $ \left(\sum_{j=1}^{t_1-1} |\lambda|^{t_1-1-j} a_j \right) \mathbb{I}_{1<t_{1}}$. Repeating this argument we get that the worst case estimation error in the interval $[t_i+1,t_{i+1}-1]$ is $ \left(\sum_{j=t_i+1}^{t_{i+1} -1} |\lambda|^{t_{i+1}-1-j} a_{ j } \right) \, \mathbb{I}_{t_i+1<t_{i+1}}$. The cost incurred by the pair $\mathbf{f}^*,\mathbf{g^*}$ is the maximum of the worst case estimation error in each interval and thus $J(\mathbf{f}^*,\mathbf{g^*}) = J^d(U^{d*}_{1:T})$ using \eqref{T_term}. 
Now, since $U^{d*}_{1:T}$ is the optimal open loop sequence it must achieve the optimal cost for Problem \ref{prob:deterministic} which is the same as the optimal cost for Problem \ref{prob:sensor} from Lemma \ref{lem:cost_equal}. Therefore, $(\mathbf{f}^*, \mathbf{g}^*)$ is globally optimal.

\section{Proof of Lemma \ref{lem:uniform_transmit}}\label{sec:lem_uniform_transmit}
Consider some open loop sequence $U^{d}_t$ and let the $K$ time instants with $U^{d}_t$ equal to $1$ be denoted by $1 \leq t_1 \leq \ldots t_K \leq T$ with the convention that $t_{K+1}= T+1, t_0 = 0$. 
Define $y_i = t_i - t_{i-1}$ for $1 \leq i \leq K+1$. We refer to $\{y_i\}_{1 \leq i \leq K+1}$ as the partition of the time horizon.
Then, $\sum_{i=1}^{K+1} y_i = T+1$. Since $K<T$, $t_{i}+1<t_{i+1}$ will hold for some $i$. Then, using the proof of Theorem \ref{thm:main}, observe that the cost incurred for a partition $\{y_i\}$ would be $(\max_i \frac{|\lambda|^{y_i -1} -1 }{|\lambda| -1}) a = (\frac{|\lambda|^{\max_i y_i -1} -1 }{|\lambda| -1}) a$ when $|\lambda| \neq 1$. We will show that $\max_i y_i$ is at least $\Delta$ for any partition. 
We first consider the case when $ \frac{T+1}{K+1} $ is not an integer. Suppose $\displaystyle\max_i y_i < \left\lceil \frac{T+1}{K+1} \right\rceil$, then  $y_i \leq \left\lfloor \frac{T+1}{K+1} \right\rfloor \, \forall i $
\begin{align}
\implies \sum_{i=1}^{K+1} y_i &\leq (K+1) \left\lfloor \frac{T+1}{K+1} \right\rfloor < T+1. \label{eq:y_i sum}
\end{align}
 \eqref{eq:y_i sum} gives a contradiction since $\sum_{i=1}^{K+1} y_i = T+1$. For the case when $ \frac{T+1}{K+1} $ is an integer, a similar contradiction can be obtained by noting that $y_i \leq  \frac{T+1}{K+1} -1 \, \forall i $. Thus, $\displaystyle\max_i y_i \geq  \left\lceil \frac{T+1}{K+1} \right\rceil = \Delta.$

Now, consider the strategy where $U_t^d = 1$ when $t = m\Delta$ for some $m \in \{1,\ldots,K\}$. Note that $l \Delta \leq T < (l+1) \Delta$ for some $1 \leq l \leq K$. It is easy to check that $\max_{i} y _i = \Delta$ for this strategy and hence it achieves the optimal cost. The proof for the case when $|\lambda| = 1$ can be easily obtained in a similar manner.

\end{document}